\documentclass[10pt]{article}

\usepackage{latexsym}
\usepackage{theorem}
\usepackage{color,graphicx}
\usepackage{amssymb}
\usepackage{amsmath}
\usepackage[OT4]{fontenc}

\newenvironment{proof}[1][Proof]
{\par\noindent{\bf #1:} }{\hspace*{\fill}\nolinebreak{$\Box$}\bigskip\par}

\newcommand{\qed}{\hspace*{\fill}\nolinebreak\ensuremath{\Box}}

\newtheorem{theorem}{Theorem}
\newtheorem{lemma}{Lemma}
\newtheorem{corollary}{Corollary}
\theorembodyfont{\upshape}
\newtheorem{definition}{Definition}

\newcommand{\Alg}{{\tt CST}}
\newcommand{\FindMinimal}{{\tt MCPS}}

\newcommand{\w}{\widetilde}
\newcommand{\jobs}{\mathcal{J}}
\newcommand{\ms}{\textup{ms}}
\newcommand{\nat}{\mathbb{N}}

\newcommand{\tds}{\textup{TDS}}
\newcommand{\csp}{\textup{CS}}
\newcommand{\csfp}{\textup{CS}_{\textup{F}}}

\newcommand{\strategy}{\mathcal{S}}
\newcommand{\collection}{\mathcal{C}}
\newcommand{\border}{\delta}
\newcommand{\clearE}{C_{\textup{E}}}

\newcommand{\sn}{\textup{\texttt{s}}}
\newcommand{\msn}{\texttt{ms}}

\newcommand{\csn}{\textup{\texttt{cs}}}
\newcommand{\mcsn}{\textup{\texttt{mcs}}}

\begin{document}

\title{Connected searching of weighted trees}

\author{Dariusz Dereniowski\thanks{Partially supported by the Foundation for Polish Science (FNP) and by the Polish Ministry of Science and Higher Education (MNiSW) grant N~N206~379337.}\\
       Department of Algorithms and System Modeling,\\
       Gda\'{n}sk University of Technology, Poland\\
       \small{deren@eti.pg.gda.pl}}

\maketitle

\begin{center}
\parbox[c]{10 cm}{
\textbf{Abstract:} In this paper we consider the problem of connected edge searching of weighted trees. It is shown that there exists a polynomial-time algorithm for finding optimal connected search strategy for bounded degree trees with arbitrary weights on the edges and vertices of the tree. The problem is NP-complete for general node-weighted trees (the weight of each edge is $1$).
}
\end{center}

\vspace{5 pt}


\section{Introduction}
\label{sec:intro}

Given a simple undirected graph $G$, a fugitive is located on an edge of $G$. The task is to design a sequence of moves of a team of searchers that results in capturing the fugitive. The fugitive is invisible for the searchers -- they can deduce the location of the fugitive only from the history of their moves; the fugitive is fast, i.e. whenever he moves, he can traverse a path of arbitrary length in the graph, as long as the path is free of searchers. Finally, the fugitive has a complete knowledge about the graph and about the strategy of the searchers, which means that he will avoid the capture as long as it is possible. The allowable moves for the searchers are, in general, placing a searcher on a vertex, removing a searcher from a vertex and sliding a searcher along an edge of $G$. An edge is \emph{clear} if it cannot contain the fugitive. Capturing the fugitive is then equivalent to clearing all the edges of $G$. The minimum number of searchers sufficient to clear the graph is the \emph{search number} of $G$, denoted by $\sn(G)$. The edge searching problem has been introduced by Parsons in~\cite{Parsons76}. The corresponding node searching problem was first studied by Kirousis and Papadimitriou in~\cite{searching_and_pebbling}. For surveys on graph searching problems see~\cite{guaranteed_graph_searching} or~\cite{searchng_and_sweeping}.

A key property of a search strategy is the monotonicity. A search is \emph{monotone} if the strategy ensures that the fugitive cannot reach an edge that has been already cleared. For most graph searching models it has been proven that there exists an optimal search strategy that is monotone. The minimum number of searchers needed to construct a monotone search strategy for $G$ is denoted by $\msn(G)$.
We have that recontamination does help for connected~\cite{sweeping_large_cliques} and connected visible search \cite{monotony_properties_connected_visible}. Moreover, the difference between $\csn(G)$ and $\mcsn(G)$ can be arbitrarily large for some graphs $G$~\cite{sweeping_large_cliques}. However, if $T$ is a tree then $\csn(T)=\mcsn(T)$~\cite{connected_weighted_trees}.

We say that a search is \emph{internal} if removing the searchers from the graph is not allowed, while for the search to be \emph{connected} we require that after each move of the searchers, the subgraph of $G$ that is clear is connected. The minimum number of searchers required for each connected search strategy of $G$ is called the \emph{connected search number} of $G$, denoted by $\csn(G)$. The corresponding monotone connected search number is denoted by $\mcsn(G)$.

Clearly $\csn(G)\geq\sn(G)$ for each graph $G$, since each connected search strategy is also a search strategy. Some upper bounds for the connected search number are known, in particular $\csn(T)\leq 2\sn(T)-2$, where $T$ is a tree~\cite{searching_not_jumping}. Connected search number is at most $O(k\log n)$, where $k$ equals the branchwidth of $G$ \cite{price_of_connectedness}. The latter implies that $\csn(G)\leq c\log n\cdot\sn(G)$, where $c$ is a fixed number, which is also a consequence of the results in \cite{connected_treewidth}. For the search with visible fugitive both search numbers are equal~\cite{connected_treewidth}.

Several algorithmic results for connected searching of special classes of graphs are know, including chordal graphs~\cite{connected_searching_chordal_graphs}, hypercubes~\cite{hypercubes,contiguous_hypercube}, a pyramid~\cite{pyramid}, chordal rings and tori~\cite{chordal_rings_tori}, or outerplanar graphs~\cite{connected_outerplanar}.
For results on searching planar graphs with small number of searchers and small number of connected components of the cleared subgraph see~\cite{Nowakowski_flooding}.

Authors in~\cite{connected_weighted_trees} provided an efficient algorithm for searching weighted trees. However, their algorithm does not always produce an optimal solution (the tree in Figure~\ref{pic:extensions} in Section~\ref{sec:searching_trees} may serve as an example), which results in an approximation algorithm. The complexity status of searching weighted trees turns out to be NP-complete, which we prove in this work.

This paper is organized as follows. In the next section we give the necessary definitions. In Section~\ref{sec:searching_trees} we analyze the basic properties of connected searching of weighted trees.\footnote{A different model of edge searching of weighted graphs, than the one considered here and in~\cite{connected_weighted_trees}, has been recently introduced in~\cite{edge_searching_weighted_graphs}.} Then, in Section~\ref{sec:bounded_degrees}, we give an algorithm for computing optimal search strategies for weighted trees. The algorithm is exponential in the maximum degree of a tree. Thus, it is designed for trees of bounded degree. Section~\ref{sec:hard} deals with the complexity of searching trees. In Subsection~\ref{subsec:reduction} we prove that finding an optimal connected search of a weighted tree is strongly NP-hard, i.e. it is NP-hard for trees with integer weight functions with polynomially (in the size of the tree) bounded values on the vertices and edges. This justifies the exponential, in general, running time of the algorithm. In order to present the proof we need a preliminary result that a special instance of scheduling time-dependent tasks is NP-complete, which is proven in Subsection~\ref{subsec:time-dependent}.

\section{Preliminaries}
\label{sec:preliminaries}

In the following we assume that all the graphs $G=(V(G),E(G),w)$ are connected, i.e. there exists a path between each pair of vertices of $G$. The sets $V(G)$ and $E(G)$ are, respectively, the vertices and the edges of $G$, while $w\colon V(G)\cup E(G)\to\nat_+$ is a weight function. ($\nat_+$ is the set of positive integers.) We start with a formal definition of the Connected Searching problem ($\csp$).

\begin{definition}
Let $k\geq 0$ be an integer. Initially all the edges of a weighted graph $G=(V(G),E(G),w)$ are \emph{contaminated}. A \emph{connected $k$-search strategy} $\strategy$ starts by placing $k$ searchers on an arbitrary \emph{starting vertex} $v_0$ of $G$. Each move of $\strategy$ consists of sliding $j\geq 1$ searchers along an edge $e\in E(G)$. If $e$ is contaminated, then we require $j\geq w(e)$, and $e$ becomes \emph{clear}. An edge $uv\in E(G)$ becomes \emph{contaminated} if there exists a contaminated edge $vy$ and less than $w(v)$ searchers occupy $v$. The subgraph that is clear has to be connected after each step of $\strategy$. After the last move of $\strategy$ all the edges of $G$ are clear.
\end{definition}

Given any strategy $\strategy$, $\sn(\strategy)$ is the number of searchers used by $\strategy$, $|\strategy|$ is the number of moves in $\strategy$ and $\strategy[i]$ is its $i$th move, $1\leq i\leq|\strategy|$. For each $i=1,\ldots,|\strategy|$, $\border(\strategy[i])$ is the set of vertices $v$, occupied by searchers at the end of move $i$, such that there exists a contaminated edge incident to $v$. We say that the vertices in $\border(\strategy[i])$ are \emph{guarded} in step $i$. In other words, if at the end of move $\strategy[i]$ there exists a vertex $v\in\border(\strategy[i])$ and less than $w(v)$ searchers occupy $v$, then a recontamination occurs.

The smallest number $k$ for which a connected $k$-search strategy $\strategy$ exists is called the \emph{connected search number} of $G$, denoted by $\csn(G)$. The minimum number of $k$ searchers such that there exists a monotone connected $k$-search for $G$ is called the \emph{monotone connected search number} of $G$, and is denoted by $\mcsn(G)$. If a (monotone) connected search strategy $\strategy$ uses ($\mcsn(G)$) $\csn(G)$ searchers, i.e. (respectively $\sn(\strategy)=\mcsn(G)$) $\sn(\strategy)=\csn(G)$, then $\strategy$ is called an \emph{optimal} (monotone) connected search strategy for $G$.

Forcing a connected search strategy to have different starting vertices results in different number of searchers required to clear a graph $G$. The problem where the starting vertex is a part of the input is denoted by $\csfp$ (Connected Searching problem with Fixed starting vertex).

The number of searchers used for guarding at the end of step $\strategy[i]$ is denoted by $|\strategy[i]|$. Note that
\[|\strategy[i]|=\sum_{v\in\border(\strategy[i])}w(v).\]
The searchers which are not used for guarding in a given step $\strategy[i]$, called \emph{free} searchers in step $i$. In particular, if more than $w(v)$ searchers occupy $v\in\border(\strategy[i])$, then $w(v)$ of them are guarding $v$, while the remaining ones are considered to be free. Free searchers can move arbitrarily along the clear edges until the next move $\strategy[i']$, $i'>i$, which clears an edge $uv$, where $u\in\border(\strategy[i])$. The move $\strategy[i']$ can be performed only if the required number of $j$ searchers (with $j'$ free searchers among them), which will slide along $uv$ in $\strategy[i']$, is at $u$. So, each move among $\strategy[i+1],\ldots,\strategy[i'-1]$ which is not necessary for gathering the $j$ searchers for clearing $uv$ in $\strategy[i']$ can be performed after $\strategy[i']$. Moreover, each set of $j'$ searchers, which are free at the end of move $\strategy[i]$, can be used to clear $uv$ in $\strategy[i']$. For this reason, we do not list the moves of sliding searchers along clear edges. Thus, due to this simplifying assumption, $|\strategy|=|E(G)|$.

We say that a strategy is \emph{partial} if it clears a subset of edges of $G$. Given a search strategy $\strategy$ for $G$, the symbol $\strategy[\preceq i]$ is used to denote the partial search strategy containing the moves $\strategy[1],\ldots,\strategy[i]$. Clearly, if $\strategy$ is connected, then $\strategy[\preceq i]$ is also connected. Given a partial search strategy $\strategy'$, we extend our notation so that $\border(\strategy')$ is the set of guarded vertices after the last move of $\strategy'$, $\border(\strategy')=\border(\strategy'[|\strategy'|])$. The symbol $\clearE(\strategy')$ denotes the set of edges cleared by a partial strategy $\strategy'$. In particular, if $\strategy$ clears $G$, then $\border(\strategy)=\emptyset$ and $\clearE(\strategy)=E(G)$.

\section{Searching trees -- basic properties}
\label{sec:searching_trees}

We are able to make several simplifying assumptions on connected search strategies once we consider the $\csp$ problem for weighted trees $T=(V(G),E(G),w)$.

In Sections~\ref{sec:bounded_degrees} and~\ref{sec:hard} we provide the algorithm for $\csfp$ problem on bounded degree trees and a polynomial-time reduction from a NP-complete problem to the $\csfp$ problem for general trees. In both cases we conclude that the corresponding result (an efficient algorithm or a polynomial-time reduction) holds for the $\csp$ problem on trees as well. We will use the symbol $\csn(T_r)$ to denote the minimum number of searchers needed to clear $T$ when $r$ is the starting vertex. Then,
\begin{equation} \label{eq:fixed_start}
\csn(T)=\min\{\csn(T_v)\colon  v\in V(T)\}.
\end{equation}
To simplify the notation, all trees $T$ are rooted at $r\in V(T)$. In the remaining part of this paper we consider the $\csfp$ problem with the starting vertex $r$.

Given a tree $T=(V(T),E(T),w)$ rooted at $r\in V(T)$, $E_v$ is the set of edges between $v$ and its descendants, $v\in V(T)$, and $T_v$ is the subtree of $T$ rooted at $v$.

For each tree $T$ it holds $\mcsn(T)=\csn(T)$~\cite{connected_weighted_trees}. Thus, in what follows each connected search strategy is monotone. As mentioned in Section~\ref{sec:preliminaries}, we only list the clearing moves of a search strategy $\strategy$, which implies $|\strategy|=|E(T)|$.

Consider a connected search strategy $\strategy$ for $T$. Let $\strategy[i]$ be a move of clearing an edge $uv$. If $v$ is a leaf, then the number of searchers that need to slide along $uv$ to clear it in step $\strategy[i]$ is $w(uv)$. When $uv$ gets clear at the end of move $\strategy[i]$, there is no need to guard $v$, which means that the searchers that reach $v$ in $\strategy[i]$ are free at the end of the move $\strategy[i]$. This holds regardless of the weight of $v$, $w(v)$. Similarly, if $u$ is a leaf, then $u=r$ and $i=1$ and it is easy to see that $w(uv)$ searchers suffice to clear $uv$, and $r$ does not have to be guarded at the end of move $\strategy[1]$. So, we may w.l.o.g. assume that
\begin{equation} \label{eq:leaf}
w(v)=1\textup{ for each leaf }v\in V(T).
\end{equation}
The number of searchers that slide along $uv$, $v$ is a son of $u$, is $\max\{w(uv),w(v)\}$ for all edges $uv$. This follows from~(\ref{eq:leaf}) when $u$ or $v$ is a leaf, while in the remaining cases $w(uv)$ searchers are needed to clear $uv$, and there have to be at least $w(v)$ searchers at $v$ at the end of $\strategy[i]$ to avoid recontamination. Thus, if the search is required to be connected and $w(v)>w(uv)$ then $w(v)-w(uv)$ searchers which are not necessary for clearing $uv$ follow along $uv$ together with the $w(uv)$ searchers that clear the edge.

Our next simplifying assumption is considering only node-weighted trees, and we argue that it does not lead to losing generality. Given a connected search strategy $\strategy$ for $T$ with starting vertex $r$, consider a move $\strategy[i]$ of clearing an edge $uv$, where $v$ is a son of $u$. At the beginning of $\strategy[i]$ the vertex $v$ is unoccupied and $u$ is guarded by $w(u)$ searchers. To clear $uv$ we need to slide $\max\{w(uv),w(v)\}$ searchers along $uv$. If $w(uv)<w(v)$, then by~(\ref{eq:leaf}) $v$ is not a leaf of $T$, which means that at the end of move $\strategy[i]$ at least $w(v)$ searchers have to occupy $v$. This means that we have to slide $w(v)$ searchers along $uv$ regardless of $w(uv)$. Thus, we may assume that if $w(uv)\leq w(v)$, then $w(uv)=w(v)$. As a result, for each edge $uv$, where $u$ is the father of $v$ we have
\begin{equation} \label{eq:edge-weights}
w(uv)\geq w(v).
\end{equation}

Consider now a new tree $T'=(V(T'),E(T'),w')$ obtained from $T$ by replacing each edge $uv$ by two edges $ux_{uv}$ and $vx_{uv}$, where $x_{uv}$ is a new vertex of $T'$ corresponding to the edge $uv$ of $T$. Let $w'(ux_{uv})=w'(vx_{uv})=1$ and $w(x_{uv})=w(uv)$ for each $uv\in E(T)$ and let $w'(v)=w(v)$ for each $v\in V(T)$. Clearly, $|E(T')|=2|E(T)|$.

For an example of all the transformations given above see Figure~\ref{pic:simplify}.
\begin{figure}[htb]
\begin{center}
\includegraphics[scale=1]{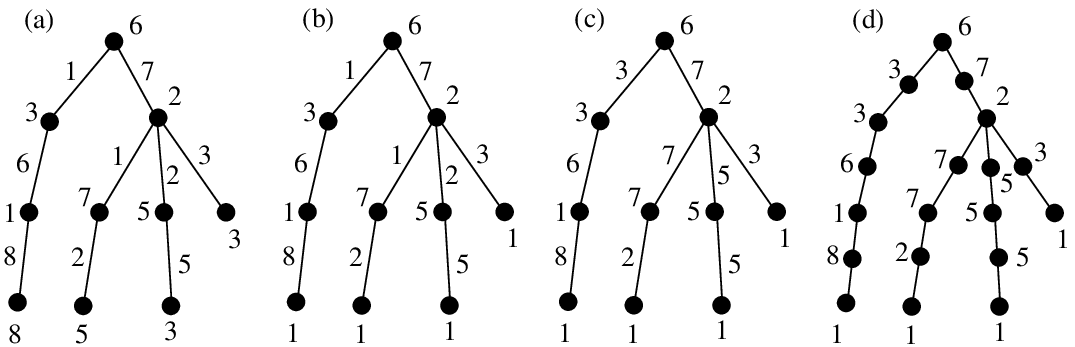}
\caption{(a) a rooted tree with node and edge weights;
         (b) the weight of each leaf is $1$;
         (c) the corresponding tree satisfying~(\ref{eq:edge-weights});
         (d) the node-weighted tree $T'$ obtained from $T$}
\label{pic:simplify}
\end{center}
\end{figure}

\begin{lemma} \label{lem:edges_unweighted}
For each $T$ and its corresponding tree $T'$, $\csn(T_r')=\csn(T_r)$ for each $r\in V(T)$.
\end{lemma}
\begin{proof}
Given a connected search strategy $\strategy$ for $T$, we create a connected search strategy $\strategy'$ for $T'$ as follows. Each move $\strategy[i]$, $1\leq i\leq|\strategy|$, clearing an edge $uv$, where $u$ is the father of $v$, is replaced by two moves $\strategy'[2i-1]$ and $\strategy[2i]$ of clearing the edges $ux_{uv}$ and $vx_{uv}$, respectively. A simple induction on the number of moves in $\strategy$ allows us to prove that $\sn(\strategy')=\sn(\strategy)$. Indeed, by~(\ref{eq:edge-weights}), clearing $uv$ in $\strategy$ requires $w(uv)$ searchers excluding the searchers used for guarding, and by the definition of $T'$, $w(uv)$ searchers are sufficient to clear $ux_{uv}$ and $vx_{uv}$ resulting in the same set of guarded vertices in $\strategy$ and $\strategy'$ after moves $\strategy[i]$ and $\strategy'[2i]$, respectively. This proves that $\csn(T_r')\leq\csn(T_r)$.

Let $\strategy'$ be a connected search strategy for $T'$. We may w.l.o.g. assume that if $\strategy'[i]$ clears an edge $ux_{uv}$, where $x_{uv}$ is a son of $u$ then, a move of clearing $vx_{uv}$ follows, because $w(v)\leq w(vx_{uv})$ by~(\ref{eq:edge-weights}). Two consecutive moves of clearing $ux_{uv}$ and $vx_{uv}$ in $\strategy'$ can be translated into clearing $uv$ in a connected search strategy which requires $w(uv)=w'(x_{uv})$ searchers. Thus, $\sn(\strategy)=\sn(\strategy')$, and consequently $\csn(T_r)\leq\csn(T_r')$. This proves that $\csn(T_r)=\csn(T_r')$.
\end{proof}

In the remaining part of this paper we assume that the weight of each edge $e\in E(T)$ is $1$.

\begin{definition}
Let $\strategy$ and $\strategy'$ be partial search strategies for $T$, where $\clearE(\strategy)\cap\clearE(\strategy')=\emptyset$. We define a search strategy $\strategy\oplus\strategy'$ as follows:
\begin{list}{}{}
\item[1.] $(\strategy\oplus\strategy')[i]=\strategy[i]$ for each $i=1,\ldots,|\strategy|$,
\item[2.] $(\strategy\oplus\strategy')[|\strategy|+i]$, $i=1,\ldots,|\strategy'|$, clears the edge cleared in the move $\strategy'[i]$, while the set of guarded vertices at the end of the move $(\strategy\oplus\strategy')[|\strategy|+i]$ is $\border((\strategy\oplus\strategy')[|\strategy|+i])=\border(\strategy'[i])\cup(\border(\strategy)\setminus X)$, where $X$ is the set of vertices initially occupied by $\strategy'$.
\end{list}
\end{definition}
In other words, $\strategy\oplus\strategy'$ clears all the edges cleared by $\strategy$ and $\strategy'$ in the order corresponding to the moves $\strategy[1],\ldots,\strategy[|\strategy|],\strategy'[1],\ldots,\strategy'[|\strategy'|]$.
Note that in particular we have that $\clearE((\strategy\oplus\strategy')[\preceq i])=\clearE(\strategy[\preceq i])$ for each $i=1,\ldots,|\strategy|$, and $\clearE((\strategy\oplus\strategy')[\preceq (|\strategy|+i)])=\clearE(\strategy)\cup\clearE(\strategy'[\preceq i])$ for each $i=1,\ldots,|\strategy'|$.
Furthermore, for $\strategy\oplus\strategy'$ to be a partial connected search starting at $r$, $\strategy$ has to be a partial connected search with starting vertex $r$, however, $\strategy'$ does not have to be connected, but the requirement is that after each step of $\strategy'$, each subgraph cleared by $\strategy'$ has to have a common vertex with $\border(\strategy)$.

\begin{definition}
Given a tree $T$ rooted at $r$, a vertex $v\in V(T)$, and an integer $k\geq 0$. We say that a partial connected $k$-search $\strategy_v$ for $T_v$, $v\in V(T)$, is $(k,v)$-\emph{minimal} if $w(\border(\strategy_v))\leq w(v)$ and $w(\border(\strategy_v))\leq w(\border(\strategy_v'))$ for each partial connected $k$-search $\strategy_v'$ for $T_v$.
\end{definition}
A strategy $\strategy_v$ is \emph{not minimal} if there exists no $k$ such that $\strategy$ is $(k,v)$-minimal.
We say that a partial connected search strategy $\strategy$ for $T_r$ can be \emph{extended} to a connected $(k,r)$-minimal search for $T_r$ if there exists a search strategy $\strategy'$ such that $\strategy\oplus\strategy'$ is a connected $(k,r)$-minimal search for $T_r$. This in particular implies that $\sn(\strategy)\leq k$. Given a $T_r$ and $E'\subseteq E(T_r)$, $T_r-E'$ is a set of maximal rooted subtrees induced by the edges in $E(T_r)\setminus E'$.

\begin{lemma} \label{lem:minimal_S_v}
A partial (not minimal) connected search strategy $\strategy$ for $T_r$ can be extended to a $(k,r)$-minimal search for $T_r$ if and only if there exist $T_v'$ (rooted at $v$) in $T-\clearE(\strategy)$ and a partial $(k-w(\border(\strategy)\setminus\{v\}),v)$-minimal connected search $\strategy_v$ for $T_v'$, such that $\strategy\oplus\strategy_v$ can be extended to a partial $(k,r)$-minimal search for $T_r$.
\end{lemma}
\begin{proof}
The ``only if'' part is obvious. To prove the ``if'' part let $\strategy\oplus\strategy_1$ be a $(k,r)$-minimal partial connected search for $T_r$. For each $v\in\border(\strategy)$ there exists a contaminated edge in $E_v$, which gives that there exists in $T_r-\clearE(\strategy)$ a nonempty subtree $T_r'$ rooted at $v$. (If all edges in $E_v$ are contaminated, then $T_v'=T_v$.) First we argue that there exist $v\in \border(\strategy)$ and a partial $(k-w(\border(\strategy)\setminus\{v\}),v)$-minimal connected search $\strategy_v$ for $T_v'$. For each $v\in\border(\strategy)$ and for each move $\strategy_1[i]$ define $B(i,v)=\border(\strategy_1[i])\cap V(T_v')$. Find minimum $l$ such that $w(B(l,v))<w(v)$ for some $v\in\border(\strategy)$. Such an integer $l$ does exist, because otherwise $w(\border(\strategy\oplus\strategy_1))\geq\border(\strategy)$ which contradicts the minimality of $\strategy\oplus\strategy_1$. Let $\strategy_v'$ be $\strategy_1$ restricted to the edges in $\clearE(\strategy_1[\preceq l])\cap E(T_v')$. By the minimality of $l$, $\strategy\oplus\strategy_v'$ uses at most $k$ searchers (which gives that $\sn(\strategy_v')\leq k-w(\border(\strategy)\setminus\{v\})$), and $w(\border(\strategy_v'))=w(B(l,v))<w(v)$. So, the set of partial $(k-w(\border(\strategy)\setminus\{v\}))$-search strategies $\strategy_v'$ for $T_v'$ satisfying $w(\border(\strategy_v'))<w(v)$ is nonempty and, by the definition, a strategy $\strategy_v$ with the minimal $w(\border(\strategy_v))$ is $(k-w(\border(\strategy)\setminus\{v\}),v)$-minimal.

We will use $\strategy_1$ to extend $\strategy\oplus\strategy_v$ to a partial $(k,r)$-minimal connected search $\strategy\oplus\strategy_v\oplus\strategy_2$ for $T_r$. To obtain $\strategy_2$ we simply remove from $\strategy_1$ all the operations of clearing the edges in $\clearE(\strategy_v)$, preserving the order of clearing the remaining edges in $\strategy_1$. One can prove that $\strategy\oplus\strategy_v\oplus\strategy_2$ is connected.

Since $w(\border(\strategy\oplus\strategy_v\oplus\strategy_2))\leq w(\border(\strategy\oplus\strategy_1))$, it remains to prove that $\sn(\strategy\oplus\strategy_v\oplus\strategy_2)\leq k$. By the definition, $\sn(\strategy\oplus\strategy_v)\leq k$, so let us consider a move $(\strategy\oplus\strategy_v\oplus\strategy_2)[i_2]$ of clearing an edge $e$, $i_2>|\strategy\oplus\strategy_v|$. Select $i_1>|\strategy|$ so that $(\strategy\oplus\strategy_1)[i_1]$ is the move of clearing $e$. Now we prove that $|(\strategy\oplus\strategy_v\oplus\strategy_2)[i_2]|\leq|(\strategy\oplus\strategy_1)[i_1]|$. Let
\begin{equation} \label{eq:U_def}
U=\border((\strategy\oplus\strategy_v\oplus\strategy_2)[i_2])\setminus \border((\strategy\oplus\strategy_1)[i_1]).
\end{equation}
In other words, $U$ is the set of vertices guarded in step $i_2$ of $\strategy\oplus\strategy_v\oplus\strategy_2$ but unguarded in step $i_1$ of $\strategy\oplus\strategy_1$. Clearly, $U\subseteq\border(\strategy_v)$. For each $u\in U$ there exists a vertex $x_u\in \border((\strategy\oplus\strategy_1)[i_1])$ on the path connecting $v$ and $u$ in $T_v'$. Let $X_U$ be the set of all such vertices $x_u$, $u\in U$. We have that
\begin{equation} \label{eq:X_U}
w(X_U)\geq w(U).
\end{equation}
To prove~(\ref{eq:X_U}) assume for a contradiction that it does not hold. Find a set $X$, with minimum $w(X)$, such that each path connecting $v$ and $u$, $u\in\border(\strategy_v)$, contains a vertex in $X$ (possibly $u$). We have $w(X)<w(\border(\strategy_v))$, because $U\subseteq\border(\strategy_v)$. Let us create $\strategy_v'$ which clears the edges in
\[\clearE(\strategy_v)\setminus\bigcup_{x\in X}E(T_{x})\]
in the same order as they are cleared in $\strategy_v$. We have $\sn(\strategy_v')\leq\sn(\strategy_v)$ and $w(\border(\strategy_v'))=w(X)<w(\border(\strategy_v))$. Thus, $\strategy_v$ is not $(k-w(\border(\strategy)\setminus\{v\}),v)$-minimal --- a contradiction, which proves~(\ref{eq:X_U}). Hence, $|(\strategy\oplus\strategy_v\oplus\strategy_2)[i_2]|\leq|(\strategy\oplus\strategy_1)[i_1]|\leq\csn(T_r)$. Since $i_2$ has been chosen arbitrarily, we have proven the thesis.
\end{proof}

As an example consider a tree in Figure~\ref{pic:extensions}(a). Assume that we start by clearing three edges $ru$, $rv$, $rw$ (in this order) and let $\strategy$ by such a partial search strategy. We have that $\sn(\strategy)=12$ and $\border(\strategy)=\{u,v,w\}$. Let us look at search strategies for selected subtrees. Denote by $\strategy_x$, $\strategy_v$, $\strategy_y$ and $\strategy_z$ search strategies for $T_x$, $T_v$, $T_y$ and $T_z$, respectively, such that the branches of the corresponding subtrees are cleared starting with the one on the left hand side, while the right branch is cleared last in all cases. They are \mbox{$(12,x)$-,} \mbox{$(8,v)$-,} \mbox{$(9,y)$-} and $(11,z)$-minimal, respectively. Also, there exist a partial $(8,u)$-minimal search $\strategy_u$ for $T_u$ with $\border(\strategy_u)=\{x\}$ (this strategy clears the two edges on the path connecting $u$ and $x$) and a partial $(8,w)$-minimal search strategy $\strategy_w$ for $T_w$, where $\border(\strategy)=\{y,z\}$ ($\strategy_w$ clear the three edges on the paths connecting $w$ and $y,z$). Suppose that we want to find a connected $12$-search strategy for $T_r$. In order to do it we extend $\strategy$. We have to find a $(12-w(\border(\strategy)\setminus\{a\}),a)$-minimal search, where $a\in\border(\strategy)$. For $a=u$ ($a=v$) we need a $(12-5,u)$-minimal (respectively $(12-6,v)$-minimal) search strategy, so $\strategy_u$ ($\strategy_v$, resp.) does not suffice. However, $\strategy_w$ is $(8,w)$-minimal and $12-w(\border(\strategy)\setminus\{w\})=12-3\geq 8$, so the moves of $\strategy_w$ as a part of $\strategy\oplus\strategy_w$ use $w(\border(\strategy)\setminus\{w\})+\sn(\strategy_w)=11$ searchers and $\border(\strategy\oplus\strategy_w)=\{u,v,y,z\}$. Now we can extend $\strategy$ by using $\strategy_u$, $\strategy_v$, $\strategy_y$ or $\strategy_z$, but only one extension, namely $\strategy\oplus\strategy_w\oplus\strategy_u$ uses no more than $12$ searchers. We have $\border(\strategy\oplus\strategy_w\oplus\strategy_u)=\{x,v,y,z\}$. The final extension (the only one possible) is $\strategy\oplus\strategy_w\oplus\strategy_u\oplus\strategy_v\oplus\strategy_y\oplus\strategy_z\oplus\strategy_x$. Note that not all minimal strategies have been listed --- there exist a $(12,w)$-minimal strategy (namely $\strategy_w\oplus\strategy_y\oplus\strategy_z$) for $T_w$ and a $(12,u)$-minimal one ($\strategy_u\oplus\strategy_x$) for $T_u$, but it is easy to check that none of those can be used to extend $\strategy$.
\begin{figure}[htb]
\begin{center}
\includegraphics[scale=1]{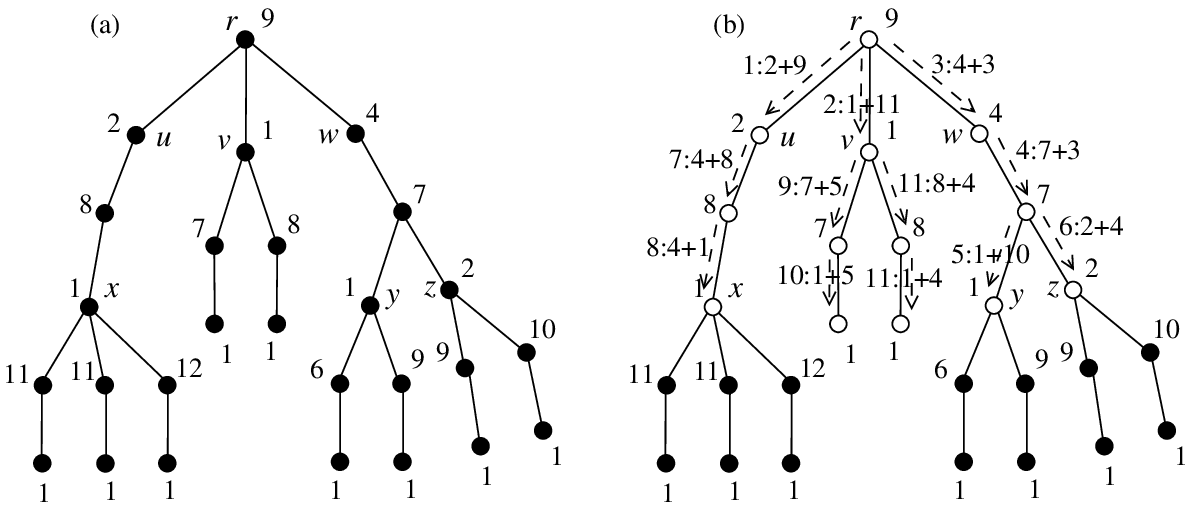}
\caption{(a) node weighted tree $T_r$; (b) $\strategy\oplus\strategy_w\oplus\strategy_u\oplus\strategy_v$}
\label{pic:extensions}
\end{center}
\end{figure}
Figure~\ref{pic:extensions}(b) depicts a partial strategy $\strategy\oplus\strategy_w\oplus\strategy_u\oplus\strategy_v$, where the dashed arrows represent the moves of the strategy. Their labels $i:c+g$ indicate the number $i$ of the corresponding clearing move, while $c$ and $g$ are, respectively, the number of searchers used for clearing and guarding in the move.

\section{Efficient algorithm for bounded-degree trees}
\label{sec:bounded_degrees}

In this section we provide a polynomial-time optimal algorithm for bounded degree trees. In an informal way, it may be described as follows. We start with placing $k$ searchers at the root $r$ of $T$. Assume that the algorithm calculated a partial search strategy $\strategy$. If $\border(\strategy)=\emptyset$ then $\strategy$ clears $T_r$ and the computation stops. Otherwise we select a vertex $v\in\border(\strategy)$ and we find a partial connected search $\strategy_v$ for $T_v$. We continue with $\strategy\oplus\strategy_v$. Note that $\strategy\oplus\strategy_v$ requires $\sn(\strategy)$ to perform $\strategy$ and then the moves of $\strategy_v$ follow, where $w(\border(\strategy)\setminus\{v\})$ searchers are used to guard the vertices that are not in $T_v$ and, in addition, $\sn(\strategy_v)$ searchers work on the subtree $T_v$. So, if $\strategy$ can be extended to a connected $k$-search for $T_r$ and we are able to find a $(k-w(\border(\strategy)\setminus\{v\}),v)$-minimal strategy $\strategy_v$, then, by Lemma~\ref{lem:minimal_S_v}, we have that $\strategy\oplus\strategy_v$ can be extended to a connected $k$-search for $T_r$. The fact that any such vertex $v$ is sufficient reduces the size of the search space for the algorithm. However, it follows immediately from the NP-completeness proof in Section~\ref{sec:hard} that finding a strategy $\strategy_v$ is intractable, unless P=NP. We point out here that Lemma~\ref{lem:minimal_S_v} will not be needed in its most general form, because we will apply it for $T_v'=T_v$, i.e. when we select a vertex $v\in\border(\strategy)$ and the corresponding search strategy $\strategy_v$, then all the edges in $E_v$ are contaminated at the end of $\strategy$.

For each $v\in V(T_r)$ a set $\collection_v$ is a global variable and will contain partial $(k,v)$-minimal connected search strategies for a subtree $T_v$, for selected values of $k$.

We start by describing a procedure, called $\FindMinimal$ (\emph{Minimal Connected Partial Strategy}), which for given integer $k$, a rooted tree $T_r$, and an ordering $rv_1,\ldots,rv_d$ of the edges incident to $r$, finds a $(k,r)$-minimal partial connected search strategy $\strategy$, which clears the edges in $E_r$ according to the given order, whenever such a strategy exists. Our final algorithm will process $T_r$ in a bottom-up fashion, so when $\FindMinimal$ is called, then for each $v\in V(T_r)\setminus\{r\}$ some $(k',v)$-minimal search strategies for $T_v$ belong to $\collection_v$ for some integers $k'$. Moreover, $w(r)$ searchers already occupy $r$ when $\FindMinimal$ starts. The procedure is as follows:
\begin{list}{}{}
\item[Step 1.] For each $i=1,\ldots,d-1$ repeat the following: (i) if $k$ searchers are sufficient to clear $rv_i$, then clear $rv_i$ as the next step of $\strategy$ and find $(k',v_i)$-minimal search $\strategy_{v_i}\in\collection_r$ with maximum $k'$, $k'\leq k-w(\border(\strategy)\setminus\{v_i\})$. If $\strategy_{v_i}$ exists, then let $\strategy:=\strategy\oplus\strategy_{v_i}$, otherwise proceed to $i+1$; (ii) if more than $k$ searchers are needed to clear $rv_i$, then return `failure'.
\item[Step 2.] Clear $rv_d$. (If $k'$ searchers are not sufficient to do it, then return `failure'.) While there exist $v\in \border(S)$ and $\strategy_v\in\collection_v$ such that $\strategy_v$ is $(k',v)$-minimal, $k'\leq k-w(\border(\strategy)\setminus\{v\})$, then $\strategy:=\strategy\oplus\strategy_v$.
\item[Step 3.] Return $\strategy$.
\end{list}

\begin{lemma} \label{lem:FindMinimal_works}
If $\strategy$ can be extended to a $(k,r)$-minimal search strategy that clears the edges in $E_r$ according to the order $\pi=(rv_1,\ldots,rv_d)$, then $\FindMinimal$ returns such a strategy.
\end{lemma}
\begin{proof}
Assume that there exists a partial $(k,r)$-minimal search strategy $\strategy_{\textup{opt}}$ clearing the edges in $E_r$ according to the order $\pi$. Let, for brevity, $\strategy_i$ denote the partial connected search strategy calculated in Steps~1-2 of $\FindMinimal$, where clearing $rv_i$ is the last move of $\strategy_i$, $i=1,\ldots,d$.

Now we use induction on $i=1,\ldots,d$ to prove that $\strategy_i$ can be extended to $(k,r)$-minimal search for $T_r$. The claim follows immediately for $i=1$, since by assumption, $\strategy_{\textup{opt}}$ starts by clearing $rv_1$. (For a connected search starting at $r$ an edge in $E_r$ has to be cleared first.) Assume that $rv_i$ has been cleared by $\strategy_i$, $i<d$. The procedure $\FindMinimal$ proceeds in Step~1 by finding a $(k-w(\border(\strategy_i)\setminus\{v_i\}),v_i)$-minimal partial connected search $\strategy_{v_i}$ for $T_{v_i}$. By Lemma~\ref{lem:minimal_S_v}, $\strategy_i\oplus\strategy_{v_i}$ can be extended to a $(k,r)$-minimal connected search for $T_r$. By the definition, there is no $v\in\border(\strategy_i\oplus\strategy_{v_i})\setminus\{r\}$ for which there exists a $(k-w(\border(\strategy_i\oplus\strategy_{v_i})\setminus\{r\}),v)$-minimal partial connected search for $T_v$. Thus, the next edge $e$ cleared by $\strategy_i\oplus\strategy_{v_i}$ must be in $E_r$. Hence, $e=rv_{i+1}$ which results in strategy $\strategy_{i+1}$.

Thus, we obtain that $\strategy_d$ can be extended to a $(k,r)$-minimal connected search for $T_r$. Then, $\FindMinimal$ finds in Step~2 a sequence of vertices $v_{d+1},\ldots,v_{d+l}$ and search strategies $\strategy_{d+1},\ldots,\strategy_{d+l}$ such that $\strategy_{d+i}$ is $(k-w(\border(\strategy_{d}\oplus\cdots\oplus\strategy_{d+i-1})\setminus\{v_{d+i}\}),v_{d+i})$-minimal and $v_{d+i}\in \border(\strategy_{d}\oplus\cdots\oplus\strategy_{d+i-1})$. By Lemma~\ref{lem:minimal_S_v}, each strategy $\strategy_{d}\oplus\cdots\oplus\strategy_{d+i}$, $i=0,\ldots,l$, can be extended to a $(k,r)$-minimal search for $T$.

Let $\strategy=\strategy_{d}\oplus\cdots\oplus\strategy_{d+l}$. We have that $\strategy$ is $(k,r)$-minimal, because otherwise, as proved above, it can be extended to a $(k,v)$-minimal search for $T_r$, and consequently, by Lemma~\ref{lem:minimal_S_v}, there exists $v\in\border(\strategy)$ and a $(k-w(\border(\strategy)\setminus\{v\}),v)$-minimal search $\strategy_{v}$ such that $\strategy\oplus\strategy_v$ can be extended to a $(k,v)$-minimal search for $T_r$, which gives a contradiction with the fact that no such vertex has been found following $v_{d+l}$ by $\FindMinimal$.
\end{proof}

Now we are ready to give a listing of the algorithm $\Alg$ (\emph{Connected Searching of a Tree}) for finding an optimal connected search strategy for a rooted tree $T_r$. This algorithm is exponential in the maximum degree of $T$, $\Delta=\max\{\deg_T(v)\colon v\in V(T)\}$.

\begin{list}{}{}
\item[Step 1.] For each son $v$ of $r$ call $\Alg(T_r)$. This step guarantees that for each $v\in V(T)\setminus\{r\}$ the collection $\collection_v$ of all minimal search strategies for $T_v$ is calculated (which is necessary also for subsequent calls of $\FindMinimal$).
\item[Step 2.] Fix a permutation $\pi=(rv_1,\ldots,rv_d)$ of the edges in $E_r$. Set $k:=1$. If Step~3 has been executed for all the $d!$ permutations $\pi$, then Exit.
\item[Step 3.] Call $\FindMinimal(k,T_r,\pi)$. If the `failure' has been returned, then increase $k$ and repeat Step~3. If a search strategy $\strategy_r$ has been returned and there is no $\strategy\in\collection_r$ such that $w(\border(\strategy))\leq w(\border(\strategy_r))$ and $\sn(\strategy)\leq\sn(\strategy_r)$ then add $\strategy_r$ to $\collection_r$ and remove from $\collection_r$ all search strategies $\strategy\neq\strategy_r$ such that $w(\border(\strategy))\geq w(\border(\strategy))$ and $\sn(\strategy)\geq\sn(\strategy_r)$. If $\border(\strategy_r)=\emptyset$ then go to Step~2 to fix the next permutation $\pi$. Otherwise increase $k$ and repeat Step~3.
\end{list}

\begin{lemma} \label{lem:finds_minimal}
Let $k$ be an integer. The set $\collection_v$ contains a partial $(k,r)$-minimal connected partial search strategy for $T_r$ whenever such a strategy exists.
\end{lemma}
\begin{proof}
We prove the lemma by induction on the number of the vertices of a tree. For a tree with one vertex the claim follows.

Let $T$ be a tree with $n>1$ vertices. By the induction hypothesis, after Step~1 of $\Alg$, the set $\collection_v$ contains a $(k',v)$-minimal connected search strategy for each $v\in V(T)\setminus\{r\}$ and for each $k'\geq 1$ whenever such a strategy exists.

Then, $\Alg$ iterates over all permutations $\pi$ of the edges in $E_r$ and for each permutation all integers $k$ are used (we stop when a strategy clearing $T_r$ has been found). Lemma~\ref{lem:FindMinimal_works} gives the thesis.
\end{proof}
Lemma~\ref{lem:finds_minimal} in particular implies, that $\Alg$ finds an optimal solution to the $\csfp$ problem, because an optimal connected search strategy $\strategy$ is $(\csn(T_r),r)$-minimal and $\border(\strategy)=\emptyset$. Now we finish this section with some complexity remarks.

\begin{lemma} \label{lem:running_time}
Given a bounded degree tree $T$, the running time of the algorithm $\Alg$ is $O(n^3\log n)$, where $n=|V(T)|$.
\end{lemma}
\begin{proof}
Denote by $\strategy_i$ the connected search strategy $\strategy$ calculated by $\FindMinimal$ for $k=i$, and for fixed $T_v$ and $\pi$. For a given permutation $\pi$ there are at most $n$ different search strategies that can be returned by $\FindMinimal$, because if $\strategy_i\neq\strategy_j$, $i<j$, then $\clearE(\strategy_i)\subsetneq\clearE(\strategy_j)$. This means that $|\collection_v|\leq\Delta!n=O(n)$. We maintain $\collection_v$ as a balanced binary search tree which gives that inserting, removing and finding search strategies takes $O(\log n)$ time. This implies $O(n\log n)$ running time of $\FindMinimal$.

As to the complexity of $\Alg$, we have that it is called $n$ times, once for each vertex. For a fixed permutation $\pi$, Step~3 of $\Alg$ is executed for at most $n$ different values of $k$. (The latter follows from the observation, that the instruction `increment $k$' in $\FindMinimal$ jumps to the next $k$ for which the next strategy found for the same $\pi$ is different, which means that at least one additional edge of $T_v$ will be cleared. The next value of $k$, for which the outcome of $\FindMinimal$ will be different, can be recorded while executing the current execution of $\FindMinimal$.) In one repetition of this step it takes $O(n\log n)$ time to execute $\FindMinimal$, and $O(n\log n)$ time to iterate over $\collection_v$ to remove unnecessary strategies from the collection. So, the running time of Step~3 of $\Alg$ is $O(n^2\log n)$, and the overall execution time of $\Alg$ is $O(n^3\log n)$.
\end{proof}

Since the algorithm solves the $\csfp$ problem, where the starting vertex is given, in order to solve the $\csp$ problem, a straightforward approach is to call $\Alg$ for each vertex of $T$ as the root and the solution is the best strategy found. However, we can reduce the running time. Let $v\in V(T)$. For different roots $r\in V(T)$ for each $v\in V(T)$ there are at most $\deg_T(v)+1$ different subtrees $T_v$ for which $\Alg$ calculates search strategies, namely each neighbor of $v$ can be its father and $v$ may be the root itself. This gives that there are in total at most $\sum_{v\in V(T)}(\deg_T(v)+1)=2|E(T)|+|V(T)|\leq 3n$ different subtrees $T_v$ to consider.
\begin{theorem} \label{thm:polynomial_for_trees}
Given a bounded degree weighted tree $T$, an optimal connected search strategy for $T$ can be computed in $O(n^3\log n)$ time, where $n=|V(T)|$.
\qed
\end{theorem}

\section{Connected searching of weighted trees is hard}
\label{sec:hard}

\subsection{Scheduling time-dependent tasks}
\label{subsec:time-dependent}

In this section we recall a problem of scheduling time-dependent (deteriorating) tasks. The execution time of a task depends on its starting time. The set of tasks is denoted by $\jobs=\{J_1,\ldots,J_n\}$. Each task $J_j\in\jobs$ is characterized by two parameters, deadline $d_j$ and running time $p_j$, which depends on $s_j$, the point of time when the execution of $J_j$ starts. The completion time of $J_j$ is $C_j=s_j+p_j$. We are interested in the single machine scheduling. A schedule $D$ is \emph{feasible} if  the completion time  $C_j$ of each task $J_j$ is not greater than its deadline, $C_j\leq d_j$, and the execution times of two different tasks do not overlap. The \emph{makespan} of a schedule $D$ is $\ms(D)=\max\{C_j\colon J_j\in\jobs\}$. Since the execution time depends on the starting point, we will write $p_j(t)$ to refer to the execution time of $J_j$ when it starts at $t\geq 0$. Observe that a schedule $D$ can be described by a permutation $\pi_D\colon\{1,\ldots,|\jobs|\}\to\jobs$, because the idle times between the execution of two consecutive tasks are not necessary for non-decreasing (in time) execution times. In the Time-Dependent Scheduling ($\tds$) problem we ask whether there exists a feasible schedule for $\jobs$. A good survey and a more detailed description of this problem can be found in~\cite{deteriorating_survey}. For a survey on scheduling problems and terminology see~\cite{Blazewicz96,Brucker_SchedulingAlgorithms}.

There are several NP-completeness results for very restricted (linear) functions for execution time of a task \cite{ChengDing03,Kubiak_deteriorating}. However, we need for the reduction described in the next subsection the $\tds$ problem instances, such that each task starts and ends at integers, which are bounded by a polynomial in the number of tasks. This property does not follow directly from the reductions in~\cite{ChengDing03,Kubiak_deteriorating}. For this reason we will prove NP-hardness of the $\tds$ problem instances having the properties we need.

We will reduce the $3$-partition problem~\cite{GareyJohnson79} to $\tds$. The former one can be stated as follows. Given a positive integer $B$, a set of integers $A=\{a_1,\ldots,a_{3m}\}$ such that $\sum_{j=1,\ldots,3m}a_j=mB$ and $B/4<a_j<B/2$ for each $j=1,\ldots,3m$, find subsets $A_1,\ldots,A_m$ of $A$ such that $A=\bigcup_{i=1,\ldots,m}A_i$, $A_i\cap A_{i'}=\emptyset$ for $i\neq i'$, and $\sum_{a_j\in A_i}a_j=B$ for each $i=1,\ldots,m$.

Now, using $B$ and $A$, we define the instance of the $\tds$ problem. Let $L=mB^3+Bm(m+1)/2$. To simplify the statements we partition the interval $[0,L]$ into intervals $I_1,\ldots,I_m$ as follows:
\begin{equation} \label{eq:intervals}
I_i=\Big[(i-1)B^3+\frac{(i-1)i}{2}B,iB^3+\frac{i(i+1)}{2}B\Big),\quad i=1,\ldots,m.
\end{equation}
We use the symbols $l_i$, $r_i$ to denote the endpoints of an interval $I_i$, i.e. $I_i=[l_i,r_i)$, $i=1,\ldots,m$.
Clearly, $\bigcup_{i=1,\ldots,m}I_i=[0,L]$ and $r_i=l_{i+1}$ for each $i=1,\ldots,m-1$. Note that the length of $I_i$ is $|I_i|=B^3+iB$ for each $i=1,\ldots,m$.

Now we define the tasks in the $\tds$ problem. For each $a_j\in A$ we introduce a task $J_j\in\jobs$ with parameters
\[d_j=L,\textup{ and }p_j(t)=ia_j\textup{ for each }t\in I_i.\]
In addition, for each $i=1,\ldots,m$ we define a task $\w{J}_i$ with the deadline $\w{d}_i$ and execution time $\w{p}_i$, where
\[\w{d}_i=l_i+B^3,\textup{ and }\w{p}_i(t)=B^3\textup{ for each }t\geq 0,\]
$i=1,\ldots,m$. Let $\w{\jobs}=\{\w{J}_1,\ldots,\w{J}_m\}$. Observe that in each schedule all tasks are executed within $[0,L]$.

For a given schedule $D$ for $\jobs\cup\w{\jobs}$, $s_j$ and $C_j$ denote, respectively, the start and completion time of $J_j\in\jobs$. Similarly, $\w{s}_i$ and $\w{C}_i$ are start and completion times of $\w{J}_i\in\w{\jobs}$. We say that a task $J$ \emph{precedes} $J'$ in a given schedule if $J$ starts earlier than $J'$.

In the next three lemmas we prove several properties of every schedule for $\jobs\cup\w{\jobs}$. Then, in Lemma~\ref{lem:reduction} we prove that there exists a schedule for $\jobs\cup\w{\jobs}$ if and only if there exists a $3$-partition for $A$ and $B$.

\begin{lemma} \label{lem:preceeding}
In each schedule $D$ for $\jobs\cup\w{\jobs}$ we have that $\w{J_i}$ precedes $\w{J}_{i+1}$ for each $i=1,\ldots,m-1$.
\end{lemma}
\begin{proof}
Suppose, for a contradiction, that the claim does not hold for $D$. Let $\w{\pi}_D$ be the permutation of tasks in $\w{\jobs}$ such that for each pair of tasks $\w{J}_i,\w{J}_{i'}\in\w{\jobs}$ we have $\w{\pi}_D^{-1}(\w{J}_i)<\w{\pi}_D^{-1}(\w{J}_{i'})$ if and only if $\pi_D^{-1}(\w{J}_i)<\pi_D^{-1}(\w{J}_{i'})$. In other words, to obtain $\w{\pi}_D$ we simply restrict $\pi_D$ to tasks in $\w{\jobs}$. Then, find the smallest index $i\in\{1,\ldots,m\}$ such that $\w{\pi}_D(i)\neq\w{J}_i$. Clearly, $\w{\pi}_D(i)=\w{J}_k$, $k>i$. We have
\[\w{C}_k\geq\w{p}_k(\w{s}_k)+\sum_{i'=1,\ldots,i-1}\w{p}_{i'}(\w{s}_{i'})=iB^3.\]
Since $\w{J}_i$ is executed in $D$ later than $\w{J}_k$, we have that
\[\w{C}_i\geq\w{C}_k+\w{p}_i(\w{s}_i)\geq(i+1)B^3>iB^3+\frac{i(i+1)}{2}B=\w{d}_i,\]
because $B^3>m^2B\geq Bi(i+1)/2$ for $i\leq m<B$. This gives the desired contradiction.
\end{proof}

Given a schedule $D$ for $\jobs\cup\w{\jobs}$, define $\w{I}_i=[\w{l}_i,\w{r}_i)=[\w{C}_i,\w{s}_{i+1})$ for $i=1,\ldots,m-1$ and let $\w{I}_m=[\w{C}_m,L)$. By Lemma~\ref{lem:preceeding}, this definition is valid and all the tasks in $\jobs$ have to be scheduled within $\bigcup_{i=1,\ldots,m}\w{I}_i$.

\begin{lemma} \label{lem:time_windows}
If $D$ is a schedule for $\jobs\cup\w{\jobs}$, then $\w{I}_i\subseteq I_i$ for each $i=1,\ldots,m$.
\end{lemma}
\begin{proof}
By the definition, $\w{C}_i=\w{l}_i$, and, by Lemma~\ref{lem:preceeding},
\begin{equation} \label{eq:left_endpoint}
\w{C}_i\geq\sum_{1\leq i'\leq i}\w{p}_{i'}(\w{s}_{i'})=iB^3\geq(i-1)B^3+\frac{(i-1)i}{2}B=l_i,\quad i=1,\ldots,m.
\end{equation}
For the right endpoint of $\w{I}_i$, $i\in\{1,\ldots,m-1\}$, we have 
\begin{equation} \label{eq:right_endpoint}
\w{r}_i=\w{s}_{i+1}\leq\w{d}_{i+1}-\w{p}_{i+1}(\w{s}_{i+1})=l_{i+1}+B^3-B^3=l_{i+1}=r_i.
\end{equation}
Since $\w{r}_m=L=r_m$, by~(\ref{eq:left_endpoint}) and~(\ref{eq:right_endpoint}) we have that $\w{l}_i\geq l_i$ and $\w{r}_i\leq r_i$, which implies $\w{I}=[\w{l}_i,\w{r}_i)\subseteq I_i$ for each $i=1,\ldots.m$.
\end{proof}

\begin{lemma} \label{lem:time_windows_len}
If $D$ is a schedule for $\jobs\cup\w{\jobs}$, then $|\w{I}_i|=iB$ for each $i=1,\ldots,m$.
\end{lemma}
\begin{proof}
We assume, for a contradiction, that the thesis does not hold for $D$.
We create a new set of tasks corresponding to $\jobs$, namely $J_j^1,\ldots,J_j^{a_j}$ are $a_j$ tasks corresponding to $J_j\in\jobs$. The set of all tasks $J_j^l$ is denoted by $\jobs'$. Note that $|\jobs|=mB$. For each $J_j^l\in\jobs'$ we define the deadline to be the same as for $J_j$, while the execution time is $p_j^l(t)=i$, where $t\in I_i$, $l=1,\ldots,a_j$. Consider a schedule $D_0$ for $\jobs'\cup\w{\jobs}$ obtained from $D$ in such a way that each task $J_j\in\jobs$ is replaces by the sequence $J_j^1,\ldots,J_j^{a_j}$. We have that a task $J_j$ executes within $\w{I}_i$ for some $i\in\{1,\ldots,m\}$, and by Lemma~\ref{lem:time_windows} $\w{I}_i\subseteq I_i$, which means that its execution time is $ia_i$. Also by Lemma~\ref{lem:time_windows} we have that the sum of execution times of $J_j^1,\ldots,J_j^{a_j}$ is $\sum_{l=1,\ldots,a_i}i=ia_i$. This in particular means that $\ms(D)=\ms(D_0)$ and all tasks in $\w{J}$ are executed in the same time intervals in both schedules.

Now we will perform a sequence of modifications of the schedule $D_0$, obtaining a sequence of schedules $D_1,D_2,\ldots,D_q$ for the set of tasks $\jobs'\cup\w{\jobs}$. We describe the first modification leading us from $D_0$ to $D_1$ and the migration from $D_p$ to $D_{p+1}$ is analogous for each $p$, $0<p<q$. In the remaining part of this proof we use symbols $\w{s}_i(D_p)$, $\w{C}_i(D_p)$, $s_i(D_p)$, $C_i(D_p)$ to distinguish the parameters of tasks which depend on a schedule $D_p$, $p\geq 0$. Consequently we write $\w{I}_i(D_p)$ since the endpoints depend on the execution time of $\w{J}_i$'s. For a task $J_j^l\in\jobs'$ its start and completion time in a schedule $D_p$ is $s_j^l(D_p)$ and $C_j^l(D_p)$, respectively. Find in $D_0$ the interval $\w{I}_i(D_0)$ such that $|\w{I}_i(D_0)|\neq iB$ and $|\w{I}_{i'}(D_0)|=i'B$ for each $i'=1,\ldots,i-1$. Such $\w{I}_i(D_0)$ does exist since we assumed for a contradiction that the thesis does not hold. Moreover, $i<m$.

If $|\w{I}_i(D_0)|>iB$ then we have that $\w{J}_{i+1}$ starts at
\begin{multline}
\w{s}_{i+1}(D_0)=iB^3+|\w{I}_i(D_0)|+\sum_{i'=1,\ldots,i-1}i'B \nonumber \\
=iB^3+|\w{I}_i(D_0)|-iB+\sum_{i'=1,\ldots,i}i'B=l_{i+1}+|\w{I}_i(D_0)|-iB. \nonumber
\end{multline}
This, however, means that $\w{J}_{i+1}$ does not finish before its deadline, $\w{C}_{i+1}(D_0)=\w{s}_{i+1}(D_0)+B^3>l_{i+1}+B^3=\w{d}_{i+1}$. So, $|\w{I}_i(D_0)|<iB$.

To obtain $D_1$, let initially $D_1=D_0$ and we apply the following modifications to $D_1$. Find in $D_1$ the task $J_j^l\in\jobs'$ which executes first in the interval $[\w{r}_i(D_1),L]$. Then, let $s_j^l(D_1)=\w{r}_i(D_1)$. Note that only tasks in $\w{\jobs}$ are executed in the interval $[\w{r}_i(D_0),s_j^l(D_0)]$. To make the schedule $D_1$ feasible, shift $i$ units to the right all tasks in $\w{\jobs}$ which are executed in $[\w{r}_i(D_0),s_j^l(D_0)]$. In the new schedule $D_1$ no two tasks overlap, because by the definition and by Lemma~\ref{lem:time_windows} the execution time of $J_j^l$ in $D_0$ is at least $(i+1)B$, while its execution time in $D_1$ is $iB$. To prove that the schedule is feasible after shifting the tasks it is enough to argue that the task $\w{J}_{i+1}$ succeeding $J_j^l$ in $D_1$ finishes before its deadline. To prove it observe that for each $i'<i$, $|\w{I}_{i'}|=i'B$ which implies that
\[\w{s}_i(D_1)=\sum_{i'=1,\ldots,i-1}(B^3+i'B)=(i-1)B^3+\frac{(i-1)i}{2}B=l_i,\]
which means that $\w{C}_i(D_1)=\w{s}_i(D_1)+B^3=l_i+B^3$, and
\[\w{C}_{i+1}(D_1)=\w{C}_i(D_1)+|\w{I}_i(D_1)|+B^3=l_i+|\w{I}_i(D_1)|+2B^3\leq l_{i+1}+B^3=\w{d}_{i+1},\]
because $|\w{I}_i(D_1)|\leq iB$. If more tasks in $\w{\jobs}$ have been shifted while computing $D_1$, then they also finish before their deadlines, because they are executed consecutively, following $\w{J}_{i+1}$. Note that there is now an idle time in $D_1$, because $s_l^j(D_1)\in I_i$ and $s_l^j(D_0)>r_i$, which by Lemma~\ref{lem:time_windows} means that the execution time of $J_j^l$ is strictly bigger in $D_0$ than in $D_1$. (Assume that the difference in execution times is $x>0$.) So, each task which succeeds $J_j^l$ in $D_0$ is executed in $D_1$ at least $x$ time units earlier, because the execution time of each task does not increase when the execution starts earlier. Consequently, $\ms(D_0)>\ms(D_1)$. Similarly, we obtain that $\ms(D_i)>\ms(D_{i+1})$ for each $i=1,\ldots,q-1$.

The schedule $D_q$ has the property that each interval $\w{I}_i(D_q)$, $i=1,\ldots,m$, is of length $iB$. So, the makespan of $D_q$ is $\ms(D_q)=mB^3+\sum_{i=1,\ldots,m}iB=mB^3+\frac{m(m+1)}{2}B=L$. Thus,
\[\ms(D)=\ms(D_0)>\ms(D_1)>\cdots>\ms(D_q)=L.\]
In particular we obtain that the makespan of $D$ exceeds $L$, while the deadline of each task in $\jobs\cup\w{\jobs}$ is at most $L$ -- a contradiction.
\end{proof}

\begin{lemma} \label{lem:reduction}
There exists a schedule for $\jobs\cup\w{\jobs}$ if and only if there exists a $3$-partition for $A$ and $B$.
\end{lemma}
\begin{proof}
Let $A_1,\ldots,A_m$ be a $3$-partition of $A$. For brevity let $\jobs_i=\{J_j\in\jobs\colon a_j\in A_i\}$. Create a schedule $D$ in such a way that \[\pi_D=(\w{J}_1,\jobs_1,\ldots,\w{J}_i,\jobs_i,\ldots,\w{J}_m,\jobs_m).\]
We use induction on $i$ to prove that the tasks in $\{\w{J}_i\}\cup\jobs_i$ are executed in time interval $I_i$. The case when $i=1$ and $i>1$ are analogous, so assume that all the tasks in $\bigcup_{1\leq i'\leq i}(\{\w{J}_{i'}\}\cup\jobs_{i'})$ are executed within $I_1\cup\cdots\cup I_i=[0,r_i]$ for some $1\leq i<m$. For $\w{J}_{i+1}\cup\jobs_{i+1}$ we have that $\w{J}_{i+1}$ is scheduled first and its execution time is $B^3$. Then, the tasks in $\jobs_{i+1}$ follow in any order. Moreover, for each $t\in I_{i+1}$ we obtain $\sum_{J_j\in\jobs_{i+1}}p_j(t)=(i+1)\sum_{a_j\in A_{i+1}}a_j=(i+1)B$, because $A_{i+1}$ is a part of the solution to the $3$-partition problem. Thus, by~(\ref{eq:intervals}), the tasks in $\{\w{J}_{i+1}\}\cup\jobs_{i+1}$ can be executed within $[r_i,r_i+B^3+(i+1)B]=[l_{i+1},l_{i+1}+B^3+(i+1)B]=I_i$.

Let $D$ be a schedule for $\jobs\cup\w{\jobs}$. By Lemma~\ref{lem:time_windows_len}, $|\w{I}_i|=iB$ for each $i=1,\ldots,m$. Let $i\in\{1,\ldots,m\}$. Since, by the definition of $\w{I}_i$'s the tasks executed within $\w{I}_i$ belong to $\jobs$ and, by Lemma~\ref{lem:time_windows}, executing $J_j$ in $\w{I}_i$ takes $ia_j$ time. Thus, for the jobs $\jobs_i\subseteq\jobs$ executed within $\w{I}_i$ we have that their total running time is $iB$, i.e. $\sum_{J_j\in\jobs_i}ia_j=iB$. So, $A_i=\{a_j:J_j\in\jobs_i\}$, $i=1,\ldots,m$, is a solution to the $3$-partition problem.
\end{proof}

\begin{theorem} \label{thm:tds_hard}
Given a set of tasks $\jobs$ with integer deadlines and integer nondecreasing (in time) execution times, the problem of deciding if there exists a feasible schedule for $\jobs$ is strongly \textup{NP}-complete.
\qed
\end{theorem}

\subsection{Reducing $\tds$ to $\csp$}
\label{subsec:reduction}

In this subsection we prove NP-hardness of $\csp$ problem. We start by reducing $\tds$ to $\csfp$, then we conclude that $\csp$ is NP-complete as well.

The instance of $\tds$ consists of a set of tasks $\jobs$, where each task $J_j\in\jobs$ has its integer deadline $d_j$ and a nondecreasing function $p_j\colon\{0,\ldots,d_j-1\}\to\nat_+$ describing the execution time. As argumented in the previous section, the integer valued functions $p_j$ imply that in each schedule $s_j$ and $C_j$ are integers, $J_j\in\jobs$, which also justifies that we may consider the values of $p_j$ only at integer points. For each $J_j\in\jobs$ let $f_j$ be the latest possible integer starting point for $J_j$, i.e. $f_j=\max\{t\in\nat\colon t+p_j(t)\leq d_j\}$. The integer $L$ is selected to be an upper bound for the length of each feasible schedule,
\begin{equation} \label{eq:L_def}
L=\max\{d_j\colon J_j\in\jobs\}.
\end{equation}

Given $\jobs$, we create a node-weighted tree $T=(V,E,w)$ rooted at $r$.
For each $J_j\in\jobs$ create a path $P_j$ with
\[V(P_j)=\{u_j^i,v_j^i\colon i=0,\ldots,f_j\},\]
\[E(P_j)=\{u_j^iv_j^i\colon i=0,\ldots,f_j\}\cup\{v_j^iu_j^{i+1}\colon i=0,\ldots,f_j-1\}.\]

The tree $T$, in addition to the vertices in $\bigcup_{J_j\in\jobs}V(P_j)$, contains the vertices $r$ and $y_j,z_j$, $j=0,\ldots,|\jobs|$.
The root $r$ is adjacent to $y_0$ and to the endpoint $u_j^{f_j}$ of each path $P_j$, $j=1,\ldots,|\jobs|$. The other endpoint of $P_j$, namely the vertex $v_j^{0}$, is adjacent to $y_j$ for each $j=1,\ldots,|\jobs|$. Finally, for each $j=1,\ldots,|\jobs|$ the vertex $y_j$ is the father of $z_j$.

The weight function $w\colon V(G)\to\nat_+$ is as follows
\begin{equation} \label{eq:weght_def1}
w(r)=2L,
\end{equation}
\begin{equation} \label{eq:weght_def3}
w(y_j)=3L,w(z_j)=1\quad j=0,\ldots,|\jobs|,
\end{equation}
\begin{equation} \label{eq:weght_def2}
w(u_j^{i})=2L-i\textup{ and }w(v_j^{i})=p_j(i)
\end{equation}
for each $j=1,\ldots,|\jobs|$, $i=0,\ldots,f_j$. Finally, let $k=4L$ be the number of available searchers.
Note that for each $u_j^i$ and $v_j^{i'}$, $0\leq i,i'\leq f_j$, it holds
\begin{equation} \label{eq:blocking}
w(u_j^i)>L\geq w(v_j^{i'}),
\end{equation}
because $f_j<L$ for each $j=1,\ldots,|\jobs|$. Other simple facts that will be useful in the following are
\begin{equation} \label{eq:u_increases}
w(u_j^0)>w(u_j^1)>\cdots>w(u_j^{f_j}),\quad j=1,\ldots,|\jobs|,
\end{equation}
\begin{equation} \label{eq:v_decreases}
w(v_j^{f_j})>w(v_j^{f_j-1})>\cdots>w(v_j^0),\quad j=1,\ldots,|\jobs|.
\end{equation}

We start by describing a search strategy $\strategy$ for $T_r$, assuming that a schedule $D$ for $\jobs$ is given:
\begin{list}{}{}
\item[Step 1:] Initially $2L$ searchers occupy $r$.
\item[Step 2:] For each $i=1,\ldots,|\jobs|$ do the following: let $J_j=\pi_D(i)$; clear the path $P_j(D)\subseteq P_j$ containing vertices $u_j^{f_j}$, $v_j^{f_j},\ldots,u_j^{s_j}$, $v_j^{s_j}$. (After this step, by~(\ref{eq:weght_def2}), $w(v_j^{s_j})=p_j(s_j)$ searchers occupy $v_j^{s_j}$ to guard it.)
\item[Step 3:] Clear the vertices $y_0$ and $z_0$.
\item[Step 4:] For each $J_j\in\jobs$ clear the path $u_j^{s_j-1}$, $v_j^{s_j-1},\ldots,u_j^{0}$, $v_j^{0}$, $y_j$, $z_j$ (after this step the subtree rooted at $v_j^{f_j}$ is clear).
\end{list}

\begin{lemma} \label{lem:C_is_a_search}
$\strategy$ is a connected search for $T$. Moreover, $\sn(\strategy)\leq k$.
\end{lemma}
\begin{proof}
It is easy to see that after each step the subtree that is clear is connected. Now we prove that the number of searchers used is at most $k$. Initially $2L$ searchers guard $r$. We prove by induction on $j=1,\ldots,|\jobs|$ that $k$ searchers suffice to clear the path $P_j(D)$ in Step~2 and the number of searchers used in $\strategy$ for guarding when the vertex $v_j^{s_j}$ becomes clear is
\begin{equation} \label{eq:induction_s}
x_j=2L+\sum_{j'\colon\pi_D^{-1}(P_{j'})\leq\pi_D^{-1}(P_j)}p_{j'}(s_{j'}).
\end{equation}

The cases when $j=1$ and $j>1$ are analogous ($x_0=2L$), so we prove it for $j$, assuming that it is true for $j-1$, $1\leq j<|\jobs|$.

By~(\ref{eq:blocking}) and~(\ref{eq:u_increases}) we obtain
\[w(u_j^{s_j})=\max\{w(v):v\in V(P_j(D))\}.\]
So, by~(\ref{eq:induction_s}), $w(u_j^{s_j})+x_j$ searchers are needed to clear $P_j(D)$. We have
\[w(u_j^{s_j})+x_j=(2L-s_j)+2L+\sum_{j'\colon\pi_D^{-1}(P_{j'})<\pi_D^{-1}(P_j)}p_{j'}(s_{j'})=4L,\]
because, by the definition of a schedule for time-dependent tasks the execution of a task $J_j$ starts immediately after the execution of the preceding task ends, which can be stated as
\[s_j=\sum_{j'\colon\pi_D^{-1}(P_{j'})<\pi_D^{-1}(P_j)}p_{j'}(s_{j'}).\]
This proves that $4L$ searchers are used in the first two steps of the algorithm. When the execution of the second step is completed, $2L$ searchers are used for guarding $r$, while for guarding the vertices $v_j^{s_j}$, $j=1,\ldots,|\jobs|$ we need
\begin{equation} \label{eq:guarding_L}
\sum_{j=1,\ldots,|\jobs|}w(v_j^{s_j})=\sum_{j=1,\ldots,|\jobs|}p_j(s_j)\leq L
\end{equation}
searchers. The last inequality follows from Equation~(\ref{eq:L_def}) and from the fact that in a valid schedule $D$ each task is completed within interval $[0,L]$. Thus, we can use $3L$ searchers to clear $y_0$, $z_0$ and then the remaining subpaths $u_j^{s_j-1}$, $v_j^{s_j-1},\ldots,u_j^{0}$, $v_j^{0}$, $y_j$, $z_j$.
\end{proof}

\begin{corollary} \label{cor:scheduling_to_search}
If there exists a valid schedule for $\jobs$, then there exists a connected $4L$-search strategy for the weighted tree $T$ rooted at $r$.
\end{corollary}

Now we prove the reverse implication, i.e. that the existence of a search strategy for $T_r$ gives a valid schedule for $\jobs$. We start with a technical lemma.
\begin{lemma} \label{lem:ry_0_cleared_last}
In each $4L$-search strategy $\strategy$ for $T_r$, $ry_0$ is the edge that is cleared last among the edges in $E_r$.
\end{lemma}
\begin{proof}
Let $\strategy[i]$ be the move of clearing $ry_0$. If at least one edge in $E_r\setminus\{ry_0\}$ is contaminated during clearing $ry_0$, the vertex $r$ has to be guarded while clearing $ry_0$. That would imply $|\strategy[i]|=w(r)+w(y_0)=5L$ --- a contradiction.
\end{proof}

\begin{lemma} \label{lem:search_to_scheduling}
If there exists a connected $4L$-search strategy $\strategy$ for the weighted tree $T_r$, then there exists a valid schedule for $\jobs$.
\end{lemma}
\begin{proof}
Given $\strategy$, define a schedule $D$, where $\pi_D(i)=J_j$ if and only if $ru_j^{f_j}$ is the $i$th cleared edge among the edges in $E_r\setminus\{r,y_0\}$. In other words, the order of clearing the edges in $E_r$ determines the order of task execution in $D$.

Let $\strategy[a_j]$ be clearing of $ru_j^{f_j}$, $j=1,\ldots,|\jobs|$, and let the move $\strategy[a_{|\jobs|+1}]$ clear $ry_0$.
By Lemma~\ref{lem:ry_0_cleared_last}, $ry_0$ is cleared last among the edges in $E_v$.

In order to prove that $D$ is valid we show two facts, namely:
\begin{list}{}{}
\item[Fact 1:] $s_j\leq f_j$ for each $j=1,\ldots,|\jobs|$.
\item[Fact 2:] The move $\strategy[a_{j+1}-1]$ clears the vertex $v_j^{s_j}$, $j=1,\ldots,|\jobs|$.
\end{list}

We use induction on $j=1,\ldots,|\jobs|$ to prove that the above facts hold.

Let $j=1$. Clearly $J_1$ starts at $s_1=0$ in $D$, which implies Fact~1 for $j=1$. We have that $2L$ searchers guard $r$ while clearing a subpath of $P_1$. Since $w(v)\leq 2L$ for each $v\in V(P_1)$, the searchers clear the whole path $P_1$, ending at $v_1^{0}=v_1^{s_1}$. Then, $y_1$ cannot be cleared, because $w(y_1)=3L$, and $w(r)=2L$ searchers occupy $r$ to guard it. So, the next move is $\strategy[a_2]$ which proves Fact~2 for $j=1$.

Assume now that Fact~1 and Fact~2 hold for some $j-1\in\{1,\ldots,|\jobs|-1\}$.

For $D$ we have $s_j=\sum_{i=1,\ldots,j-1}p_i(s_i)$. By the induction hypothesis (Fact~2) we have that the number of searchers used to guard vertices in subtrees rooted at $u_1^{f_1},\ldots,u_{j-1}^{f_{j-1}}$ is $\sum_{i=1,\ldots,j-1}w(v_i^{s_i})$. By~(\ref{eq:weght_def2}), $w(v_i^{s_i})=p_i(s_i)$, which implies that $2L+w(v)+\sum_{i=1,\ldots,j-1}p_i(s_i)=2L+w(v)+s_j$ is the number of searchers used while clearing $v\in V(P_j)$. In particular, the number of searchers used to clear $u_j^{f_j}$ is $2L+2L-f_j+s_j$. Since $\strategy$ uses $4L$ searchers, $s_j\leq f_j$ which proves Fact~1.

In the move $\strategy[a_j]$ we clear $ru_j^{f_j}$ and then the searchers clear partially the subtree rooted at $u_j^{f_j}$, ending by clearing a vertex $v_j^{x}$, $0\leq x\leq f_j$ and then the move $\strategy[a_{j+1}]$ follows. ($y_j$ cannot be cleared when $r$ is guarded, because $w(y_j)=3L$. Moreover, the search does not stop at a vertex $u_j^i$, because by~(\ref{eq:blocking}) it is possible to continue by clearing $v_j^i$ for each $i=0,\ldots,f_j$.)

If $x<s_j$ then, in particular, the vertex $u_j^{s_j-1}$ has been cleared, while $2L+s_j$ searchers are used to guard $r$ and $v_i^{s_i}$, $i=1,\ldots,j-1$. By~(\ref{eq:weght_def2}), $w(u_j^{s_j-1})=2L-s_j+1$. So, the total number of searchers used while clearing $u_j^{s_j-1}$ is $2L+s_j+2L-s_j+1>4L$ --- a contradiction.

If $x>s_j$, then we can clear $v_j^xu_j^{x-1}$, because as before $2L+s_j$ searchers are used for guarding and $w(u_j^{x-1})=2L-(x-1)$ additional searchers clear $u_j^{x-1}$, which means that the number of searchers in use is $4L+s_j-x+1\leq 4L$. Then, by~(\ref{eq:blocking}), we can clear $v_j^{x-1}$.

By Fact~1, $s_j\leq f_j$, for each task $J_j\in\jobs$, which means that $C_j\leq f_j+p_j(s_j)\leq d_j$. This proves that $D$ is valid.
\end{proof}

The $\csfp$ is clearly in NP, and the reduction is polynomial in $n$, which gives us the theorem.
\begin{theorem} \label{thm:csf_hard}
Given a weighted tree $T$ rooted at $r$ and an integer $k\geq 0$, deciding whether $\csn(T_r)\leq k$ is \textup{NP}-complete.
\qed
\end{theorem}
Let $T_r=(V(T),E(T),w)$ and $k$ be an input to the $\csfp$ problem. There exists a connected $k$-search strategy for $T_r$ if and only if there exists a connected $(2k)$-search strategy for $T_r^2=(V(T),E(T),2w)$ (we double the weights of the vertices in $T_r$). Take three copies of $T_r^2$, add a vertex $r'$, which will be the root of $T_r'$, and let the roots of the trees $T_r^2$ be the sons of $r'$. We have that $\csn(T_r')=2k+1$. Moreover, if $\strategy'$ is a connected $(2k+1)$-search strategy for $T_r'$ then regardless of the starting vertex of $\strategy'$, the strategy is forced to clear one of the subtrees $T_r^2$ in $T_r'$ by starting at $r$ and using $2k$ searchers. This leads to the following
\begin{corollary} \label{cor:cs_hard}
The problem of connected searching of weighted trees is strongly \textup{NP}-hard.
\qed
\end{corollary}

\section{Conclusions}
\label{sec:conclusions}

This paper presents a polynomial-time algorithm for finding optimal connected search strategies of a bounded degree trees with any weights on the edges and vertices of the tree. On the other hand, the corresponding decision problem is NP-complete for arbitrary trees with restricted weight functions $w$, where $w(e)=1$ for each edge $e$ and $w(v)$ is bounded by a polynomial in $n$, where $n$ is the number of vertices of the input tree.

One of the interesting open problems is the existence of `good' approximations for finding connected search strategies for trees. Note that the bound $\csn(T)\leq 2\sn(T)$ \cite{connected_weighted_trees} does not yield an approximation algorithm since no algorithms for searching weighted trees are known.

\bibliographystyle{plain}
\bibliography{search}
\end{document}